\documentclass[12pt,a4paper]{article}	
\usepackage{graphicx}
\usepackage{amsthm}
\usepackage {amsfonts} 
\usepackage {amssymb} 
\usepackage{amsmath}

\newtheorem{theorem}{Theorem}[section]
\newtheorem{definition}{Definition}[section]
\newtheorem{lemma}{Lemma}[section]

\newtheorem{corollary}{Corollary}[section]

\begin{document}
\title{Reconstruction Conjecture for Graphs Isomorphic to Cube of a Tree}

\author{  {S. K. Gupta} \and {Akash Khandelwal} }  
%

\maketitle

\begin{abstract}
This paper proves the reconstruction conjecture for graphs which are isomorphic to the cube of a tree. The proof uses the reconstructibility of trees from their peripheral vertex deleted subgraphs. The main result follows from (i) characterization of the cube of a tree (ii) recognizability of the cube of a tree (iii) uniqueness of tree as a cube root of a graph $G$, except when $G$ is a complete graph (iv) reconstructibility of trees from their peripheral vertex deleted subgraphs.
\end{abstract}


\section{Introduction}
\subsection{ The Reconstruction Conjecture (RC)}
\paragraph*{} The Reconstruction Conjecture (RC) is one of the most celebrated unsolved problems in Discrete Mathematics and Combinatorics circles. It was first conjectured by Ulam and Kelly in 1941 as stated in the survey paper by Bondy \cite{b1}. 

\subsubsection{ Original Definition}
\paragraph*{} Ulam \cite{u1} states the following problem:
\begin{quote}
``Suppose that in two sets $A$, $B$; each of $n$ elements, there is defined a distance function $\rho$ for every pair of distinct points, with values either $1$ or $2$ and $\rho(x,x) =0$. Assume that for every subset of $n-1$ points of $A$; there exists an  isometric system of $n-1$ points of $B$, and that the number of distinct subsets isometric to any given subset of $n-1$ points  is same in $A$ as in $B$. Are $A$ and $B$ isometric?''
\end{quote}

\subsubsection{ Modified Definition of the Graph Reconstruction Conjecture}
Reconstruction Conjecture can be restated as: 
\begin{quote}
``A simple finite graph $G$ with at least three points can be reconstructed uniquely (up to isomorphism) from its collection of vertex deleted subgraphs $G_i$.''  
\end{quote}
This conjecture was termed by Harary  \cite{h3},  a ``graphical disease'', along with the 4-Color Conjecture and the characterization of Hamiltonian graphs. The term ``diseases'' comes from the fact that such problems can be formulated very easily and concisely, and most identified diseases are understandable to undergraduates. They are highly contagious, thereby attracting the attention of both professionals and layman mathematicians.

\paragraph*{} The reconstruction problems provide a fascinating study of the structure of graphs. The identification of structure of a graph is the first step in its reconstruction. We can determine various invariants of a graph from its subgraphs, which in turn tell us about the structure of the graph.

\subsection{Basic Terminologies} 

\paragraph*{} The key terms in the paper are introduced below. For terms not defined here, we shall use the terminology followed in Harary \cite{h3}. 

\begin{definition} [Deck of a Graph and its Cards]
 Any graph $G$ has a vertex set $V(G)$ and an edge set $E(G)$. A card is any vertex-deleted-subgraph of $G$, with $G_i$ representing the unlabelled subgraph of $G$ with the $i^{th}$ vertex and its coincident edges removed. The deck of the graph $G$ is the multiset of all cards of $G$.   
\end{definition}

\begin{definition} [$k$-periphery of a subtree]
 Given a tree $T$ with vertex set $V$ and an arbitrary subtree $T_s$, the distance of $v$ $\in$ $V$ from $T_s$ is defined to be length of smallest path connecting $v$ to some $v'$ $\in$ $T_s$. The k-periphery of $T_s$ is defined to be a set of vertices as a distance $k$ from $T_s$. 
\end{definition}

\begin{definition}[Peripheral Vertices of a Graph] 
The eccentricity $\varepsilon_G(v)$ of a vertex $v$ in a graph $G$ is the maximum distance from $v$ to any other vertex. Vertices with maximum eccentricity are called peripheral vertices. 
\end{definition}

\begin{definition} [Power of a Graph]
 Let $G$ be a graph on $p$ points $v_1$, $v_2$,..., $v_p$. The k-th power of $G$, denoted by $P_k(G)$, is a graph on points $u_1$,$u_2$,..., $u_p$ where $u_i$ and $u_j$ ($i \neq j$) are adjacent if and only if $v_i$ and $v_j$ are at distance at most $k$ in $G$. We also call $G$ to be a k-th root of $P_k(G)$.  
\end{definition}

\paragraph* {}In general, a graph may have more than one $k$th root. The uniqueness of tree as a square root of a graph has been proven independently by Ross and Harary \cite{r2} and a simpler proof using the reconstruction conjecture, has been given by Gupta \cite{s1}\cite{s2}. The uniqueness of a tree a a cube root of a graph has been established by Yerra et al. \cite{a1} In Section 3, we use a different approach to show the uniqueness of tree as a cube root of a graph $G$, except when $G$ is a complete graph, in which case $G$ will not have a unique tree root. Further, Yerra et al. \cite{a1} showed that for any $n \ge 4 $, there exist non-isomorphic trees $T_1$ and $T_2$ such that $T_1^n \cong  T_2^n$. 

\begin{definition} [End Deleted Tree]
  The end deleted tree $\xi$ of a given tree $T$ is defined to be the tree obtained by deletion of all the leaf node of $T$. 
\end{definition}

\begin{definition} [Weighted Tree]
\label{WT}
 A weighted tree is a tree having weights associated with every vertex. The weight on vertex $i$ represents the number of branches emanating from vertex $i$ in tree $T$ [Fig. \ref{tree_cube},\ref{weighted_tree}]. Any tree $T$ is equivalent to a weighted tree having weights associated with every vertex of $\xi$, the end deleted tree of $T$.  
\end{definition}
 
\begin{definition}[$i$th order leaf nodes]
 The set of $i$th order leaf nodes $L_i$ is defined to be that set of vertices which are leaf nodes of i-times end deleted tree of $T$. i.e., a vertex $v \in  L_i$ if $v$ is a leaf node in i-times deleted tree. In the base case, $L_0$ is the set of leaf nodes of $T$. 
\end{definition}
 
\begin{definition} [Distance between two edges] 
 If $e_1\equiv\{u_1,u_2\}$ and $e_2\equiv\{v_1,v_2\}$ are two edges in the graph then the distance $d(e_1,e_2)$ between them is defined to be $n+1$, where $n = min\{d(u_1,v_1),d(u_2,v_1),d(u_1,v_2),d(u_2,v_2)\}$. 
\end{definition}

\begin{definition} [Distance between an edge and a vertex]
 If $\{u_1,u_2\}$ is an edge $e$ and $v$ is a vertex in a Graph $G$, then the distance $d(e,v)$ between $e$ and $v$ is defined as $d(e,v) = min\{d(u_1,v),d(u_2,v)\}$.  
\end{definition}
 
\begin{definition} [$k$-span of a vertex] - Let $k$ be a natural number. If $v$ is a vertex then the $k$-span of vertex $v$, $S(v,k)$ is defined to be the set of all vertices as a distance upto $k$ from $v$, i.e.,
\begin{center} $S(v,k) = \{u|d(u,v) \le k\}$ \end{center}
\end{definition}

\begin{definition} [Span of an edge]
 Let $k$ be a natural number. If $e = \{v_1,v_2\}$ is an edge, then the $k$-span of edge $e$, $S(e,k)$ is defined to be the set union of $k$-spans of its end-points, i.e.,
\begin{center} $S(e,k) = S(v_1,k) \cup S(v_2,k)$ \end{center}
\end{definition}
 
\subsection{ Discussion about the Problem }

The statement of the reconstruction conjecture excludes the trivial graph $K1$, graphs on two vertices and infinite graphs \cite{b1}\cite{n1}. The deck of graphs on two points, i.e. $K2$ and $K2'$, have a pair of $K1$s comprising each of their decks but the graphs are non-isomorphic. For every infinte cardinal $\alpha$, there exists a graph with $\alpha$ edges which is not uniquely reconstructible from its family of edge deleted subgraphs \cite{c1}.
Apart from these two exceptions which prohibit the conjecture from encompassing all graphs, unique reconstructibility is conjectured for all other graphs.

 One of the ways for tackling the RC is known as the reconstructive approach, and is followed in many of the proofs of the conjecture for specific classes.  While reconstructing a class of graphs using this approach, the problem of reconstruction partitions into two subproblems, namely recognition: showing that membership in that class is determined by the deck, and weak reconstruction: showing that no two non-isomorphic members of the class have the same deck.

 The reconstruction conjecture has been proved for trees by Kelly \cite{k1}, and squares of trees by  Gupta \cite{s1}\cite{s2}. Apart from these, the conjecture has been proved for a number of graph classes such as unicyclic graphs \cite{m2}, regular graphs \cite{n2} and disconnected graphs \cite{h1}. Though the problem can be stated very simply,yet due to a lack of a nice set of characterizing invariants, it has still not been proven for very important classes of graphs like bipartite graphs and planar graphs. For further study of this conjecture, the reader is referred to surveys by Bondy \cite{b1} and Harary \cite{h2}.

\section{ Reconstruction Conjecture For Cube of Trees }

\subsection{\label{2.1} Overview of Proof Technique}


Section 2 lists basic properties of cubes of trees. In Section 3, a characterization of cubes of trees is given. It also shows uniqueness of tree as a cube root of a graph $G$, except when $G$ is a complete graph. Section 4 proves  recognizability and weak reconstruction of graphs isomorphic to cubes of trees, utilizing reconstructibility of trees from their peripheral vertex deleted subgraphs.

\subsection{\label{2.2} Properties of Third Power }
Listed below are few properties of third power of a graph:

\begin{lemma}
\label{Lemma 2.1}
  Let $e = \{v_1,v_2\}$ be an edge. If $v_1,v_2$ $\in$  $L_{i+1}$, $i \ge 0$ then the subgraph which is the $1$-span of $e$ is a clique $\kappa$.  Any Graph $G(\cong P_3(T))$ has cliques only of this type. The edge $e$ is defined as a clique edge and clique $\kappa$ is said to be centered about $e$.
\end{lemma}
\begin{proof}
 This lemma follows directly from Lemma 3.1.2.1 in \cite{s3}.
\end{proof}

\begin{definition} [Clique Distance]
 The distance $d(S_1,S_2)$ between two cliques $S_1$ and $S_2$ in $G$ is defined to be the distance between clique edges of $S_1$
 and $S_2$. 
\end{definition}

\begin{definition} [Terminal Edge]
 An edge $e$ = $\{u,v\}$ is called a terminal edge if atleast one of $deg(u)$ and $deg(v)$ is 1.
\end{definition}

 \begin{definition} [$k$th order Terminal edges]  
The terminal edges of $k$-times end deleted tree of tree $T$ are called $k$th order terminal edges of $T$ .
 \end{definition}

\begin{definition} [Terminal Clique]
 A clique $S$ is said to be a terminal clique if $S$ is the $k$-span of edge $e$ = $\{v1,v2\}$ where $k = \frac{n - 1}{2}$ and either  $v_1$ $\in$ $L_k$ or $v_2$ $\in$ $L_k$. 
\end{definition}

\begin{lemma}
 \label{Lemma 2.2}
 $S$ is a terminal clique of graph $G$ iff there exists a unique clique $S'$ such that $\forall v$ $\in$ $S$, either $v$ $\in$ $S$ or $v$ $\in$ any clique other than $S$. Further, the clique edge of $S$ is adjacent to the clique edge of $S'$ .
\end{lemma}
\begin{proof}
 This lemma follows directly from Theorem 3.2.2.1 in \cite{s3}.
\end{proof}

\begin{lemma}
\label{Lemma 2.3} 
If $v$ is a terminal vertex in $G(\cong P_3(T))$ , then $v$ is a leaf node in $T$ .
\end{lemma}
\begin{proof}
 This lemma follows directly from Lemma 3.2.2.2 in \cite{s3}.
\end{proof}

\begin{lemma}
\label{Lemma 2.4}  If $V_T$ = $\{v_1 , v_2 , v_3 . . . v_k \}$ is the set of terminal vertices of a graph $G(\cong P_3(T))$ and $V_s$ is any subset of $V_T$ then there exists a tree $T_1$ such that $G-V_s$ $\cong P_3(T_1)$.
\end{lemma}
\begin{proof}
 This lemma follows directly from Lemma 3.2.2.3 in \cite{s3}.
\end{proof}

\begin{definition} [$k$th order terminal cliques]
 If $G(\cong P_3(T))$ is a graph and $V_T$ the set of its terminal vertices then the terminal cliques of $G'(\cong G-V_T)$ are defined to be terminal cliques of $1$st order. The $k$th order terminal cliques of $G$ can be obtained by extending this to $k$-times terminal vertices deleted graph of $G$ . 
\end{definition}

\begin{definition} [Tree of Cliques] 
\label{TOC}
 A tree of cliques is that subgraph $T'$ of $T$ of which every edge forms a clique in $G(\cong P_3(T))$. 
\end{definition}

\begin{lemma}
\label{Lemma 2.5}
 The tree of cliques $T'$ of any graph $G(\cong P_3(T))$ is the end deleted tree of T. 
\end{lemma}
\begin{proof}
 This lemma follows directly from Lemma 3.3.2.1 in \cite{s3}.
\end{proof}

\begin{figure*} 
	  \centering 
           \includegraphics[width=10cm]{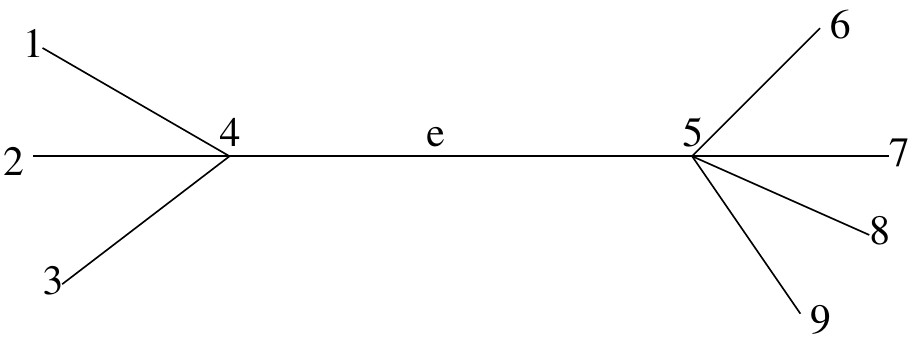}
	   \caption{\small{Tree $T$}}
	   \label{tree_cube}
\end{figure*} 

\begin{figure*}
	  \centering 
           \includegraphics[width=5cm]{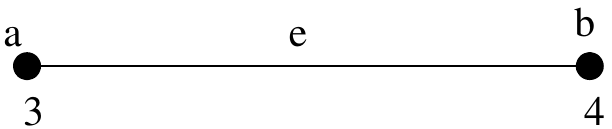}
	   \caption{\small{Weighted Tree $\tau$ equivalent to Tree $T$(Fig.~\ref{tree_cube}). }} 
	   \label{weighted_tree}
\end{figure*}

\section{Characterization of Third Power of a Tree}
\paragraph* {} Consider the formation of cliques in $P_3(T)$. Any clique can be seen centered around an edge with branches emanating from both end points. 

By Definition~\ref{WT}, there is a unique 1 to 1 mapping between any tree and its weighted tree representation. That is,  any tree $T$ is equivalent to its weighted end-deleted tree $\xi$. For example, consider the following two scenarios: (i) In the Fig. \ref{tree_cube} shown, the vertices $1$,$2$,...,$9$ form a single clique in $P_3(T)$. (ii) An equivalent representation of this tree is a weighted tree as shown in Fig. \ref{weighted_tree}. Vertices labelled $4$ and $5$ in Fig.~\ref{tree_cube} correspond to vertices a and b in Fig.~\ref{weighted_tree} with weights 3 and 5 respectively. The weights in the tree can be visualized as count of branches emanating from the corresponding vertices. Thus both (i) and (ii) have the same third power.

\begin{lemma}
\label {Lemma 3.1}
 For any Tree $T$, the end-deleted tree $\xi$ is isomorphic to the tree of cliques consisting of the edges forming cliques in $G(\cong P_3(T))$. 
\end{lemma}
\begin{proof}
This lemma follows directly from Definition~\ref{TOC} and Lemma~\ref{Lemma 2.5}.
\end{proof}


\begin{theorem}
\label{Theorem 3.1}
Let $T$ be a tree. Then $v_i$ is an end point of $T$ if and only if $P_3(T_{v_i})$ is isomorphic to ${(P_3(T))}_{v_i}$, where $G_{v_i}$ represents the vertex deleted graph obtained after removing $v_i$ from the graph $G$ .
\end{theorem}
\begin{proof}
 Let $v_i$ be a end point of $T$. Any edge $v_lv_m$ in $P_3(T_{v_i})$ is also present in $(P_3(T))_{v_i}$, where subscripts are used to indicate vertex deleted graphs. Now since there does not exist any path $v_lv_iv_m$ in $T$, there are no two points $v_l$ and $v_m$ which are adjacent in $P_3(T_{v_i})$  but not adjacent in $(P_3(T))_{v_i}$.

 For the second part, let $T$ be a tree and $v_i$ be a point of it such that $P_3(T_{v_i})$ is isomorphic to ${(P_3(T))}_{v_i}$. ${(P_3(T))}_{v_i}$ is connected for All $v_i$ but $P_3(T_{v_i})$ is connected only if $v_i$ is an end point.
\end{proof}

\begin{corollary}
\label{Corollary 3.1}
Let $T$ be a tree. Then $v_i$ is an end point of $T$ if and only if $P_3(T_{v_i})$ is isomorphic to ${P_3(T)}_{v_i}$.
\end{corollary}
\begin{proof}
It follows directly from Theorem~\ref{Theorem 3.1}.
\end{proof}

This theorem establishes a one-to-one mapping between Trees and their Third Powers. It follows the approach suggested by Yerra et al. \cite{a1}.

\begin{theorem}
\label{Theorem 3.2}
A tree $T$ can be uniquely determined from its third power,$P_3(T)(\not\cong$ $K_p$).
\end{theorem}
\begin{proof}
If $P_3(T)$ is the complete graph $K_p$, then $T$ could be any tree of diameter less than $4$. We shall prove this theorem by induction on $|T|$ for $P_3(T)\not\cong K_p$ where $|T|$ is the number of vertices in $T$. 

The hypothesis of the theorem is true for $|T|=1$ and $|T|=2$ trivially. Assume it to be true for $|T|\le r$. Let $|T|=r+1$. Consider the set $\mathbb{S}$ of point deleted subgraphs of $P_3(T)$. Gupta et al. \cite{s3}, in their discussion of characterization of power of trees show that it is possible to select a subset $\mathbb{M}$ from $\mathbb{S}$ consisting of those subgraphs which are cubes of some trees. $T_{v_i}$ is a tree only when $v_i$ is and end point and in this case, from Corollary~\ref{Corollary 3.1} we have,
\begin{center}
$P_3(T_{v_i})\cong{(P_3(T))}_{v_i}$
\end{center}
 So $\mathbb{M}$ is precisely the set ${(P_3(T))}_{v_i}$ where $v_i$ is an end point of $T$ from Theorem~\ref{Theorem 3.1}. By assumption of this theorem, $T_{v_i}$ can be uniquely determined as $|T_{v_i}|=r$. The result of this theorem now follows by induction as a tree is uniquely reconstructible from end point deleted subgraphs (\cite{h4}). 
\end{proof}



\section{Recognition and Weak Reconstruction}


Harary et al. \cite{h4} have shown that trees are reconstructible from their end vertex deleted subgraphs. In our approach, we shall use this reconstruction approach as a black box $\bar{B}$. Given the set of end vertex deleted subgraphs, $\bar{B}$ will uniquely return the tree.  In case the inputted deck does not belong to a tree, the blackbox outputs an error. Let $\mathbb{C}$ denote the class consisting of all graphs isomorphic to third power of some tree. 

\begin{lemma}
\label{Lemma 4.1}
$\mathbb{C}$ is weakly reconstructible.
\end{lemma}
\begin{proof}
We are given a set $\mathbb{S}$ of subgraphs $G_1$,$G_2$,...,$G_n$, known to be the deck of $G\in \mathbb{C}$. We have to reconstruct $G$ uniquely. Using the characterization of tree powers as discussed in Gupta et al. \cite{s3}, it is possible to select a subset $\mathbb{M}$ from $\mathbb{S}$ consisting of those subgraphs which are cubes of some tree. $T_{v_i}$ is a tree only when $v_i$ is and end point and in this case, from Corollary~\ref{Corollary 3.1} we have, $P_3(T_{v_i})\cong{(P_3(T))}_{v_i}$. So $\mathbb{M}$ is precisely the set ${(P_3(T))}_{v_i}$ where $v_i$ is an end point of $T$ from Theorem~\ref{Theorem 3.1}.

Using the set $\mathbb{M}$ and the black box $\bar{B}$ and given the fact that original deck corresponds to some member of $\mathbb{C}$, we can reconstruct $T$ and then $G \equiv P_3(T)$ uniquely.  Due to the unique reconstruction, we can conclude that no two non-isomorphic members of the $\mathbb{C}$ have the same deck, hence $\mathbb{C}$ is weakly reconstructible.
\end{proof}

\begin{lemma}
\label{Lemma 4.2}
$\mathbb{C}$ is recognizable.
\end{lemma}
\begin{proof}
We are given a set $\mathbb{S}$ of subgraphs $G_1$,$G_2$,...,$G_n$. In order for $\mathbb{C}$ to be recognizable, we have to give a boolean answer to the question: Does the deck $\mathbb{S}$ corresponds to a graph in $\mathbb{C}$? We consider both the cases in the paragraphs below.
 
 If $\mathbb{S}$ indeed corresponds to a graph in $\mathbb{C}$, we are guaranteed to obtain the unique reconstruction $G \in \mathbb{C}$ by using a strategy employed in the proof of Lemma~\ref{Lemma 4.1}. It can be verified using Deck Checking Algorithm \cite{k2}, asserting whether $\mathbb{S}$ resulted from the reconstructed $G$, and ``true'' is returned as the boolean answer. 
 
 Now consider the other case, where deck $\mathbb{S}$ didn't correspond to a graph in $\mathbb{C}$. On using the black box $\bar{B}$ over $\mathbb{M}$, we have two subcases: it will either give an error, or return a tree $T_x$. In case of error, we return ``false'' directly. In the other subcase, we obtain $G_x = P_3(T_x)$. On using the Deck Checking Algorithm \cite{k2}, $\mathbb{S}$ will not match the reconstructed $G_x$ since $\mathbb{S}$ didn't correspond to a graph in $\mathbb{C}$, so again we can return ``false''.
\end{proof}

%
%

\begin{lemma}
\label {Lemma 4.3}
RC is true for class of graphs isomorphic to cube of a tree, except for complete graphs. 
\end{lemma}
\begin{proof}
This follows directly from Lemma~\ref{Lemma 4.1} and Lemma~\ref{Lemma 4.2}
\end{proof}

\paragraph* {}
An alternate result shows that reconstruction conjecture holds trivially for complete graphs \cite{m3}. The following result follows:

\begin{theorem}
\label{Theorem 4.1} 
RC is true for $\mathbb{C}$, the class of graphs isomorphic to cube of a tree.
\end{theorem}

\section{Conclusion and Future Work}
\paragraph*{} Trees were proven to be reconstructible by Kelly \cite{k1} and squares of trees by Gupta \cite{s1}\cite{s2}. In this paper, we have proved the conjecture for graphs isomorphic to cube of a tree. It would be interesting to prove the conjecture for higher powers of trees. 

 As discussed in Yerra et al \cite{a1}, for any $n \ge 4 $, there exist non-isomorphic trees $T_1$ and $T_2$ such that $P_n(T_1) \cong  P_n(T_2)$. Thus the uniqueness argument no longer holds while proving the class of Graphs isomorphic to fourth(or higher) powers. Thus proving RC for such classes of graphs requires different approach. In general, we'd like to prove RC for Graphs isomorphic to any power $n$ of a tree.

\end{document}